\documentclass[12pt,oneside,english]{amsart}
\usepackage[latin9]{inputenc}
\usepackage{geometry}
\usepackage{textcomp}
\usepackage{amsbsy}
\usepackage{amstext}
\usepackage{amsthm}
\usepackage{amssymb}
\usepackage{stackrel}
\usepackage{graphicx}
\usepackage{setspace}
\makeatletter

\DeclareFontEncoding{LGR}{}{}
\DeclareRobustCommand{\greektext}{%
  \fontencoding{LGR}\selectfont\def\encodingdefault{LGR}}
\DeclareRobustCommand{\textgreek}[1]{\leavevmode{\greektext #1}}
\ProvideTextCommand{\~}{LGR}[1]{\char126#1}

\numberwithin{equation}{section}
\theoremstyle{plain}
\newtheorem{thm}{\protect\theoremname}
\theoremstyle{definition}
\newtheorem{example}[thm]{\protect\examplename}
\ifx\proof\undefined
\newenvironment{proof}[1][\protect\proofname]{\par
	\normalfont\topsep6\p@\@plus6\p@\relax
	\trivlist
	\itemindent\parindent
	\item[\hskip\labelsep\scshape #1]\ignorespaces
}{%
	\endtrivlist\@endpefalse
}
\providecommand{\proofname}{Proof}
\fi
\theoremstyle{definition}
\newtheorem{defn}[thm]{\protect\definitionname}
\theoremstyle{plain}
\newtheorem{cor}[thm]{\protect\corollaryname}
\theoremstyle{plain}
\newtheorem{lem}[thm]{\protect\lemmaname}
\theoremstyle{remark}
\newtheorem{rem}[thm]{\protect\remarkname}

\usepackage{lscape}
\usepackage{lineno}

\makeatother

\usepackage{babel}
\providecommand{\corollaryname}{Corollary}
\providecommand{\definitionname}{Definition}
\providecommand{\examplename}{Example}
\providecommand{\lemmaname}{Lemma}
\providecommand{\remarkname}{Remark}
\providecommand{\theoremname}{Theorem}

\begin{document}
\vspace{0.5cm}

\title{Short title: Oscillatory Stationary Population Identity}

\maketitle
\textbf{\Large{}Full Title:}{\Large{} On the three properties of stationary
populations and knotting with non-stationary populations}{\Large\par}

(To appear in \emph{Bulletin of Mathematical Biology, }Springer)

\begin{center}

\textbf{Arni S.R. Srinivasa Rao}\footnote{Laboratory for Theory and Mathematical Modeling, Department of Medicine
- Division of Infectious Diseases, Division of Epidemiology, Medical
College of Georgia, Department of Mathematics, Augusta University,
1120, 15th Street, AE 1015 Augusta, GA, 30912, USA, Tel: +1-706-721-3786
(office). Email: arrao@augusta.edu (corresponding author). }\textbf{ }and\textbf{ James R. Carey}\footnote{Department of Entomology, University of California, Davis 95616, USA,
and Center on the Economics and Demography of Aging, University of
California, Berkeley. Email: jrcarey@ucdavis.edu}

\end{center}

\vspace{1.0cm}
\begin{abstract}
A population is considered stationary if the growth rate is zero and
the age structure is constant. It thus follows that a population is
considered non-stationary if either its growth rate is non-zero and/or
its age structure is non-constant. We propose three properties that
are related to the stationary population identity (SPI) of population
biology by connecting it with stationary populations and non-stationary
populations which are approaching stationarity. One of these important
properties is that SPI can be applied to partition a population into
stationary and non-stationary components. These properties provide
deeper insights into cohort formation in real-world populations and
the length of the duration for which stationary and non-stationary
conditions hold. The new concepts are based on the time gap between
the occurrence of stationary and non-stationary populations within
the SPI framework that we refer to as\emph{ Oscillatory }SPI\emph{
}and the\emph{ Amplitude of }SPI\emph{.}

$ $
\end{abstract}

\keywords{Key words: stationary population identity, Oscillatory properties,
functional knots, PDEs}

\subjclass[2000]{AMS Subject class: 92D25, 60H35}

\section{\textbf{Stationary population identity: History and inspirations
from biological experiments}}

Stationary Population Identity (SPI) is about equality of two quantities:
one is obtained from the age-distribution of a stationary population
and the other is obtained from the remaining years to live (or remaining
time to live) of these individuals. This equality which is closely
associated with the concept of the life table (a mathematical model
to represent age-specific mortality in a population) can be expressed
in several other ways. Let $X$ be the set of elements representing
the proportions of populations at each age of a stationary population
at time $t$ and let $Y$ be the set of elements representing the
remaining number of years (or remaining time units) left to live at
each age, then SPI holds imply, 

\begin{equation}
X=Y.\label{eq:X=00003DY}
\end{equation}

In a strict sense $X$ consists of distinct elements and $Y$ consists
of distinct elements. Let us take an element in $X$, say $p$. Then,
there exists an age in the stationary population at which the proportion
of the population to the total population is $p$. If the equation
(\ref{eq:X=00003DY}) is true, then that guarantees that one of the
elements of $Y$ is also $p$. 

The equation (\ref{eq:X=00003DY}) is true in population life tables
which are stationary in nature. 

Introduced to the demography literature by Brouard (\cite{Bouard,Brouard's Book})
using French life tables and to the population biology literature
by Muller, Carey and their colleagues (\cite{MullerCarey2004,Carey2012,CareySilverRao})
using survival patterns of captive cohorts of insects, stationary
population identity (SPI) is expressed as $f_{1}(a)=f_{2}(a)$, where
$f_{1}(a)$ the fraction of individuals who are captured at age $\text{\textquotedblleft\ensuremath{a}\textquotedblright}$
(out of total population) is equal to $f_{2}(a)$ the proportion of
individuals who have a remaining time units left to die (see Figure
\ref{Stationary-population-color}). Although SPI is observed in populations
that are stationary (replacement-level growth), the vast majority
of populations for both humans and non-human species are both non-stationary
and non-stable (changing growth rate and/or age structure). All the
relevant definitions used in this paper are provided in Table \ref{tab:Key-Definitions-Used}.
\begin{example}
Consider individuals between the ages of 70 to 90 years as a stationary
subpopulation of the larger stationary population depicted in Figure
\ref{Stationary-population-color}. According to the Stationary Population
Identity (SPI), if 14.8\% of the population are in this 70 to 90 year
old subpopulation, then there also exists another sub-population of
the same number (and percentage) of individuals who have between 70
and 90 years remaining (diagonal shaded area from 0 to 30 years in
Figure \ref{Stationary-population-color}). 
\end{example}
In this article, we prove several new theoretical aspects of stationary
and non-stationary populations while understanding the implications
of SPI. Three prominent of them are listed below:

(i) Populations consist of both stationary and non-stationary components
(Theorem 2),

(ii) Stationary subpopulations of a total populations also possess
stationary components (Theorem 3), 

(iii) Population that is transiently stationary over a finite or an
infinite interval can be joined with non-stationary populations (oscillatory
property) (Theorems 6 and 7).

\begin{figure}
\includegraphics[scale=0.5]{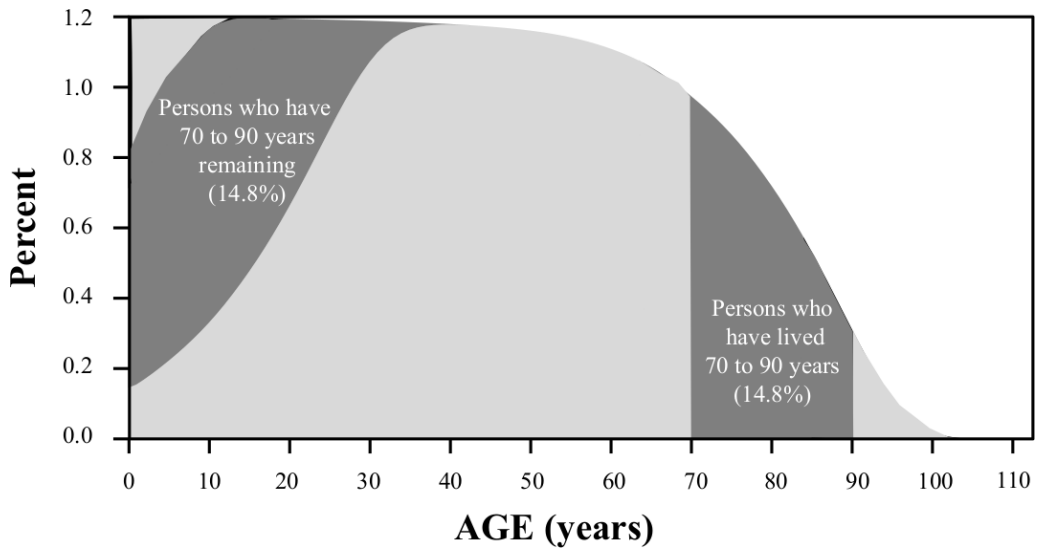}

\caption{\label{Stationary-population-color}Illustration of the stationary
population identity in which LL equals LR. Graphic is based on the
U.S. 2006 female life table in which 14.8\% of the life table population
falls between 70 and 90 years of age (i.e., life-lived). This percentage
is identical to the percentage of individuals in this same hypothetical
stationary population who have from 70 to 90 years remaining (i.e.,
life-left).}
\end{figure}

Discovery of the SPI by Carey, originally referred to as Carey's Equality
(\cite{Vaupel}, \cite{Carey'schapter2019}) but now referred to as
SPI after the revelation that Brouard's earlier papers also documented
this identity (see \cite{CareySilverRao}), was an outcome of a 10-
year, U.S. National Institute on Aging-funded research program directed
by biodemographer James R. Carey designed to study aging in the wild.
This led to the identification of the relationship between population
age structure and post-capture life spans of individuals through the
use of a simple four-age class life table (Table 1 in \cite{MullerCarey2004}),
the results of which were formulated mathematically for cases involving
both stationary (i.e., SPI) and the non-stationary (with reference
life tables) populations (see subsections on pp126-128 in \cite{MullerCarey2004}).
Because of the importance to basic ecology and particularly to medical
entomology where the older arthropod vectors (e.g., mosquitoes) have
the highest likelihood of disease transmission, a great deal of effort
has been invested in developing various technologies for estimating
the age of individual insects including physiological \cite{Detinova},
biochemical \cite{Gerade} and genetic \cite{cook,cook2008,Carey2012,Anonimous}
methods. 

The analytical evolution of this SPI continued with the publication
of its proof, first as a demographic relationship between the life
lived (LL) and life remaining (LR) \cite{Vaupel} and then as a theorem
and generalization \cite{RaoCarey2015}. A major mathematical breakthrough
came in stationary population literature, when Rao and Carey \cite{RaoCarey2015}
stated a new theorem using original ideas of stationary population
principles (Carey - Rao Theorem on stationary population identity)
through constructing arguments based on graphs and set-theoretic principles
and based on two criteria that they stated on \textquoteleft LL\textquoteright{}
and \textquoteleft LR\textquoteright . The general concepts of the
life table identify and its extension as an applied model (i.e., integration
of reference life table information) have been used to estimate age
structure and thus to gain insights into population aging in wild
populations including effects of truncation studies \cite{Rao-Carey2019HOST},
of fruit flies \cite{Carey2008,Carey2012b}, butterflies \cite{Molleman},
and mosquitoes \cite{Papa}. In light of the theoretical and analytical
properties of SPI and its use as a foundation for developing models
for estimating age structure in real-world insect populations, we
believe that continuing to explore the mathematical properties of
this identity has the potential to make new and original contributions
to the demographic literature. Thus for the non-stationary and non-life
table populations, the role of SPI needs thorough investigation. 

\begin{table}
\[
\begin{array}{cc}
\underline{\text{Item}} & \underline{\text{Definition}}\\
\\
\\
\text{SPI } & \text{\text{Two sets \ensuremath{X} and }}Y\text{ such that}X=Y,\text{where }\\
\text{(Stationary Population Identity)} & \text{\ensuremath{X} is the set of elements representing the }\\
 & \text{ proportions of populationsat each age of a }\\
 & \text{ stationary population at time \ensuremath{t} and }Y\text{ is the set}\\
 & \text{ of elements representing the remaining number of}\\
 & \text{ years (or remaining time units) left to live at each age.}\\
\\
\text{} & \text{ SPI does not hold}\text{ When }X\neq Y\text{ in the above definition,}\\
 & \text{we say that SPI does not hold.}\\
\\
\text{Knot} & \text{Joining of two simultaneous sub-intervals of time}\\
 & \text{within [\ensuremath{t_{0}},\ensuremath{t_{\omega}}), where in one sub-interval SPI holds, }\\
 & \text{and other sub-interval SPI does not hold}\\
 & \text{ }\\
 & \text{}\\
\text{Predominantly } & \text{}\text{If a population during a time interval [\ensuremath{t_{0}},\ensuremath{t_{\omega}})}\\
\text{a Stationary Population} & \text{ satisfies the inequality}\\
\text{} & [t_{0},t_{1})\cup\stackrel[M,i=1]{\infty}{\bigcup}[\delta_{i},t_{i+1})>\stackrel[N,i=1]{\infty}{\bigcup}[t_{i},\delta_{i})\\
\text{} & \text{where SPI holds within }\\
 & M=\left\{ [t_{0},t_{1}),[\delta_{1},t_{2}),[\delta_{2},t_{3}),\cdots\right\} \\
 & \text{and SPI does not holds within}\\
 & N=\left\{ [t_{1},\delta_{1}),[t_{2},\delta_{2}),\cdots\right\} \\
 & \text{and intervals in }M\text{ and }N\text{ form partitions}\\
 & \text{of }[t_{0},t_{\omega}).\\
\\
\text{Oscillatory SPI} & \text{SPI is oscillatory on the set \ensuremath{M}}\\
 & \text{ with uniform amplitude if }\\
 & \left\{ t_{1}=\frac{t_{0}+\delta_{1}}{2},t_{2}=\frac{\delta_{1}+\delta_{2}}{2},t_{3}=\frac{\delta_{2}+\delta_{3}}{2},...\right\} \\
\\
\\
\end{array}
\]

\caption{\label{tab:Key-Definitions-Used}Key Definitions Used in the Paper}

\end{table}

\section{\textbf{Stationary and Non-Stationary populations}}

While exploring the deeper insights of SPI, we realized that this
property can be helpful in knotting (read joining together) the concepts
of stationary and non-stationary populations such that these two populations
are formed on mutually-exclusive time intervals. A knot here we mean,
joining of two simultaneous sub-intervals of time within $[t_{0},t_{\omega})$,
where in one sub-interval SPI holds, and other sub-interval SPI does
not hold. The main advantages of such a theoretical visualization
of side-by-side occurrence of stationary and non- stationary populations
are to keep our framework of SPI as flexible as possible such that
realistic population dynamics are captured with respect to deviation
from stationarity. Mathematically, these mutually exclusive concepts
allow us to cut with knots the continuous interval on which we study
simultaneous occurrences of these two types of populations. Our constructions
in this article show that SPI property generates these knots on the
continuous interval. Demarcation lines on an interval between stationary
and non- stationary populations can then be visualized as dynamic.
These demarcations (or boundary) lines led us to a novel concept within
the SPI which we term \emph{Oscillatory SPI} (O-SPI). In this case,
the knots indicate the beginning of either stationary or non-stationary
populations and allow us to introduce another term that we refer to
on a continuous interval as the \emph{amplitude of the SPI}. 

For a predominantly stationary population (see Table \ref{tab:Key-Definitions-Used}
and the Definition \ref{def: predominant Stationary}) during an interval
$[t_{0},t_{\omega})$, we can imagine that there exists a disjoint
covering of intervals (a sub-collection of intervals, say $M$, in
which SPI is true and other sub-collection of intervals, say, $N$,
in which SPI is not true), such that, 

\[
\left(\bigcup_{C\in M}C\right)\cup\left(\bigcup_{C'\in N}C'\right)
\]
equals $[t_{0},t_{\omega})$. The interval $[t_{0},t_{\omega})$ is
visualized as a the union of two partitions, one which form SPI and
other does not. See \cite{Davis-book,Kelly-book} for concepts related
to disjoint covering. The partition which form the identity is associated
with stationary population and other one is associated non-stationary
populations. If $``-"$ indicates the minus symbol, then, the SPI
is true in $[t_{0},t_{\omega})-\bigcup_{C'\in N}C'$ and not true
in $[t_{0},t_{\omega})-\bigcup_{C\in M}C.$ The value $[t_{0},t_{\omega})-\bigcup_{C'\in N}C'$
indicates the interval $[t_{0},t_{\omega})$ minus the intervals $\bigcup_{C'\in N}C',$i.e.,
if an element $x$ belongs $[t_{0},t_{\omega})-\bigcup_{C'\in N}C',$
then $x$ belongs to $[t_{0},t_{\omega})$ but $x$ does not to the
union of intervals $\bigcup_{C'\in N}C'.$ Similarly, the meaning
of $[t_{0},t_{\omega})-\bigcup_{C\in M}C$ can be interpreted. We
develop an idea which we call \emph{uniform amplitude }of SPI when
equality such as (\ref{eq:equality-amplitude}) is true

\begin{equation}
[t_{0},t_{\omega})-\bigcup_{\boldsymbol{C}\in M}C=[t_{0},t_{\omega})-\bigcup_{C'\in N}C',\label{eq:equality-amplitude}
\end{equation}
and together $C=C'$ holds for each simultaneous $C\in M$ and $C'\in N.$
However, we develop these ideas on finite sets. Later we will see
that the set $T$ in (\ref{set T}) 

\begin{equation}
T=\left\{ [t_{0},t_{1}),[t_{1},\delta_{1}),[\delta_{1},t_{2}),\cdots[t_{k},\delta_{k-1}),[\delta_{k},t_{k+1}]\right\} \label{set T}
\end{equation}
is a partition of $[t_{0},t_{k+1}]$, where $[t_{0},t_{k+1}]\subset[t_{0},t_{\omega}),$
such that each element of $[t_{0},t_{k+1}]$ lie in exactly in one
interval in (\ref{set T}). We will also see in the Appendix that
the set $\left\{ I,J\right\} $ for the two intervals $I$, $J\subset[t_{0},t_{k+1}]$
as a partition of $[t_{0},t_{k+1}]$. Inasmuch as SPI connects these
two properties in stationary populations, it follows that connecting
them in non-stationary populations is a logical next step.

Let $\Omega$ be the size of the captive cohort such that $\Omega$
is an infinite subset or a very large finite subset of non-negative
integers. Let $c_{i}$ be the age at capture and $d_{i}$ be the age
at death of $i^{th}$ individual, where $d_{i}>c_{i}$ for each $i\in\Omega$.
Here, $d_{i}-c_{i}$ is the follow-up length or post-capture LL by
$i^{th}$ individual. 
\begin{thm}
\label{thm1}If a population is stationary then the SPI holds, but
when\textbf{ $f_{1}(a)=f_{2}(a)$ }does not hold for\textbf{ }every
age $``a"$ in a population then that population could be partitioned
into stationary and non-stationary components\textbf{.}
\end{thm}
\begin{proof}
Idea: To prove the first part, we need to prove that if the population
from which the captive cohort drawn is stationary then that follows
SPI. For the second part, we first assume that$f_{1}(a)=f_{2}(a)$
is\textbf{ }not true for every age\textbf{ $``a"$, }and then we try
to prove that the captive cohort $\Omega$ formed could be\textbf{
}partitioned into stationary population and non-stationary population\textbf{
}components\textbf{. }

We assume a very large number of individuals are captured at all possible
ages (need not be integer valued) and no two individuals have same
age at capture. We also assume that: (1) there will be a distinct
value of duration of LR (i.e., remaining life to be lived after capture)
corresponding to the each captured individual; and (2) one\textbf{
}of the values of the remaining LR is identical to exactly one of
the values of the age at capture. Let $\mu(\overline{c})$ and $\mu(\overline{d-c})$
be the average age at capture and average age of remaining length
of post-capture life for the individuals in $J$, respectively, and
$c_{1}\neq c_{2}\neq...\neq c_{k}\neq...$ and $d_{1}-c_{1}\neq d_{2}-c_{2}\neq...\neq d_{k}-c_{k}\neq...$.,
then we have 
\begin{equation}
\mu(\overline{c})=\frac{\Sigma_{i}c_{i}}{\left|\Omega\right|}\mbox{ and }\mu(\overline{d-c})=\frac{\Sigma_{i}(d_{i}-c_{i})}{\left|\Omega\right|}.\label{eq:mu(c)=00003Dmu(d-c)}
\end{equation}
Suppose $S=\left\{ s_{1},s_{2},...\right\} $, where $s_{i}=d_{i}-c_{i}$
for all $i\in\Omega.$ We can arrange elements of the set $S$ in
a decreasing order. To do this, we set $s'_{1}=\max\left\{ s_{1},s_{2},...\right\} .$
Let $S_{1}=S-\left\{ s'_{1}\right\} ,$ where $S_{1}$ is the set
of elements in $S$ after $s'_{1}$ is removed. Let $s'_{2}=\max\left\{ S_{1}\right\} .$
We can continue to obtain maximum values, such that $S_{i+1}=S_{i}-\left\{ s'_{i+1}\right\} ,$
where $s'_{i+1}=\max\left\{ S_{i}\right\} $ for $i=1,2,\cdots$.
Let $T=\left\{ \left(1,s'_{1}\right),\left(2,s'_{2}\right),\cdots\right\} .$
The graph drawn through the co-ordinates of $T$ is a decreasing function.
These kind of constructions for the information of LR after capture
was originally used in \cite{RaoCarey2015}. When $s'_{i}$ is equal
to the corresponding individual's age at capture for all $i\in J$
then the distribution of captured age is equal to the distribution
of duration of the LR after capture. When $s'_{j}$ is not equal to
the corresponding individual's age at capture for all $j\in\Omega_{1}$
for $\Omega_{1}\subset\Omega$, and $s'_{i}$ is equal to the corresponding
individual's age at capture for all $i\in\Omega$ and $i\notin\Omega_{1}$.
Then with a finite permutations of rearrangement of the elements in
$\Omega_{1}$, we can match the set, $T'=\left\{ s'_{1},s'_{2},\cdots\right\} $
with $C,$ the set of decreasing values of captured ages, such that
$T'=C.$ With this construction explained, for an individual captured
at the\textbf{ }age\textbf{ }$\text{\textquotedblleft\ensuremath{a}\textquotedblright}$
in $C$ (i.e., $a$ is an element in $C$) the value of the (element
in $T'$) is exactly\textbf{ }$\text{\textquotedblleft\ensuremath{a}\textquotedblright}$
which is the remaining LR. 

Suppose there are one or more than one individual of the same age
at the time of capture. $\Omega$ is now sum of partitions of individuals,
where each partition represents number of individuals who are captured
at the same age. Let $c_{p}^{q}$ be the $q^{th}$ individual captured
aged $p$ and $s_{p}^{q}$ be the remaining LR for the $q^{th}$ individual
who was captured at age $p$ for $p>0$ and $q=1,2,\cdots,n_{p}$
$\left(n_{p}\in\mathbb{Z}^{+}\right).$ We assume that for each of
the $s_{p}^{q}$ there is a corresponding value $c_{p}^{q}$ which
could be within the same age $p$ or in other captured age. That is,
if 
\[
C=\left\{ c_{p}^{q}:\mbox{ }c_{p}^{q},p,\in\mathbb{R}^{+},q\in\mathbb{Z}^{+}\mbox{ and }1\leq q\leq n_{p}\right\} 
\]
 and

\[
S=\left\{ s_{p}^{q}:\mbox{ }s_{p}^{q},p\in\mathbb{R}^{+},q\in\mathbb{Z}^{+}\mbox{ and }1\leq q\leq n_{p}\right\} ,
\]
then for each $y\in S$ there is a corresponding element $x\in C.$
The following property is assured:

\begin{equation}
\frac{\int_{0}^{\infty}\left(\sum_{q=0}^{n_{p}}c_{p}^{q}\right)dp}{\left|J\right|}=\frac{\int_{0}^{\infty}\left(\sum_{q=0}^{n_{p}}s_{p}^{q}\right)dp}{\left|J\right|}.\label{eq:Carey's Equality}
\end{equation}

For more details on the type of logic and arguments provided above,
see the Carey-Rao Theorem and proof \textbf{\cite{RaoCarey2015}},
which introduced these set of arguments. Conversely, suppose for a
finite population, let\textbf{ $f_{1}(a)=f_{2}(a)$ }for the ages
$a_{1},a_{2},\cdots,a_{k}$ (without any order) and\textbf{ $f_{1}(a)\neq f_{2}(a)$}
for ages $a_{k+1},a_{k+2},\cdots,a_{n}$ (without any order) and no
two individuals are of same age. This implies, there will be two vectors\textbf{
$V_{1}$ }and $V_{2}$ based on the rule that SPI is true or not,
which are given by,

\[
V_{1}=\left[\begin{array}{c}
f_{1}(a_{1})=f_{2}(a_{1})\\
f_{1}(a_{2})=f_{2}(a_{2})\\
\vdots\\
f_{1}(a_{k})=f_{2}(a_{k})
\end{array}\right]\mbox{ and }\begin{array}{ccc}
V_{2} & = & \left[\begin{array}{c}
f_{1}(a_{k+1})\neq f_{2}(a_{k+1})\\
f_{1}(a_{k+2})\neq f_{2}(a_{k+2})\\
\vdots\\
f_{1}(a_{n})\neq f_{2}(a_{n})
\end{array}\right].\end{array}
\]

The sub-population corresponding to $V_{1}$ forms stationary population
and the sub-population corresponding to $V_{2}$ forms non-stationary
population.\textbf{ }Due to $V_{2}$ the average value of the LL by
$a_{1},a_{2},\cdots,a_{n}$ are not equal to their average remaining
value, which will lead population to be non-stationary. 
\end{proof}
\begin{thm}
\label{thm2}Suppose SPI holds for a population, then a) SPI also
holds for all stationary sub-populations of the original population
and b) SPI need not hold for all non-stationary sub-populations of
the original population. 
\end{thm}
\begin{proof}
a) Let $P$ be a stationary population and let $P_{1},P_{2},...,P_{k}$
are disjoint stationary sub-populations and $Q_{1},Q_{2},...,Q_{l}$
are disjoint non-stationary sub-populations of $P$ such that

\begin{equation}
\bigcup_{i=1}^{k}P_{i}\cup\bigcup_{j=1}^{l}Q_{j}=P.\label{eq:PuQ=00003DP}
\end{equation}
If $P_{i}$ for each $i$ is a stationary population, then SPI holds
within the each $P_{i}.$ 

b) Suppose we partition population into a disjoint collection of stationary
and non-stationary sub-populations as in (a), then SPI need not hold
in an arbitrarily chosen $Q_{j}$, because, for an arbitrarily chosen
$Q_{j}$ the intrinsic growth rates could be very high and the population
could be younger such that the proportion of population at age $a$
years not equal to the proportion of population who have $a$ years
remaining. 
\end{proof}
\begin{figure}
\includegraphics[scale=0.7]{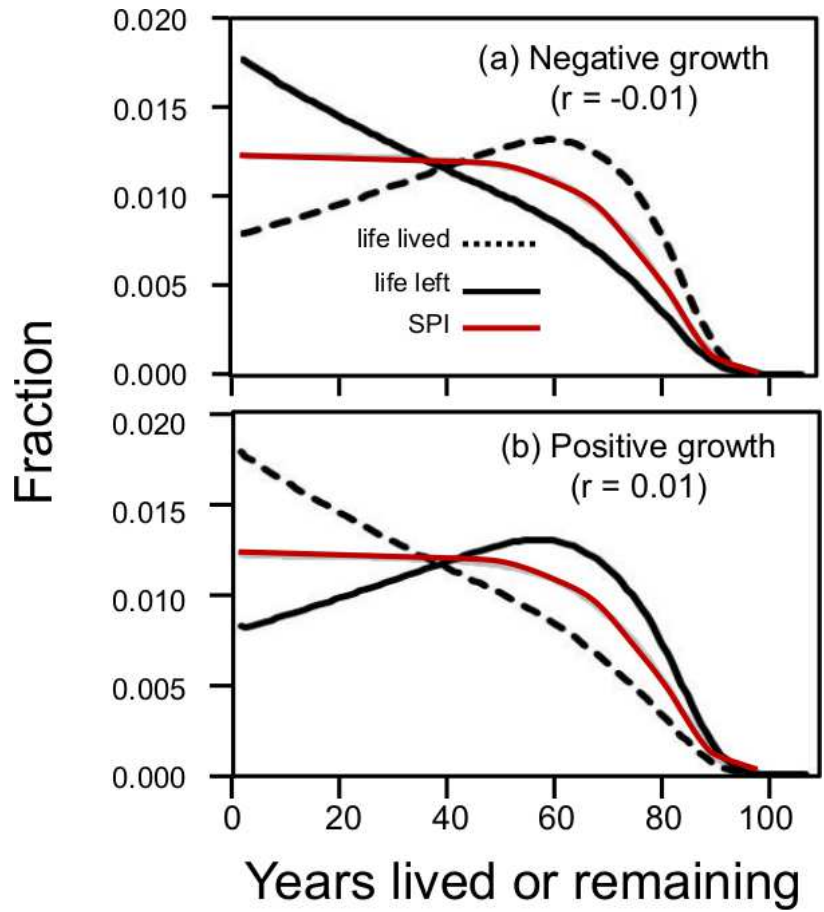}

\caption{\label{fig2}Relationship of the fraction of individuals who have
lived x years relative to the fraction who have x years left to live
for populations with either negative (a) or positive (b) growth rates.
For reference the red curves shows the equivalency of LL and left
for stationary populations (zero growth). }

\end{figure}

\begin{figure}
\includegraphics[scale=0.5]{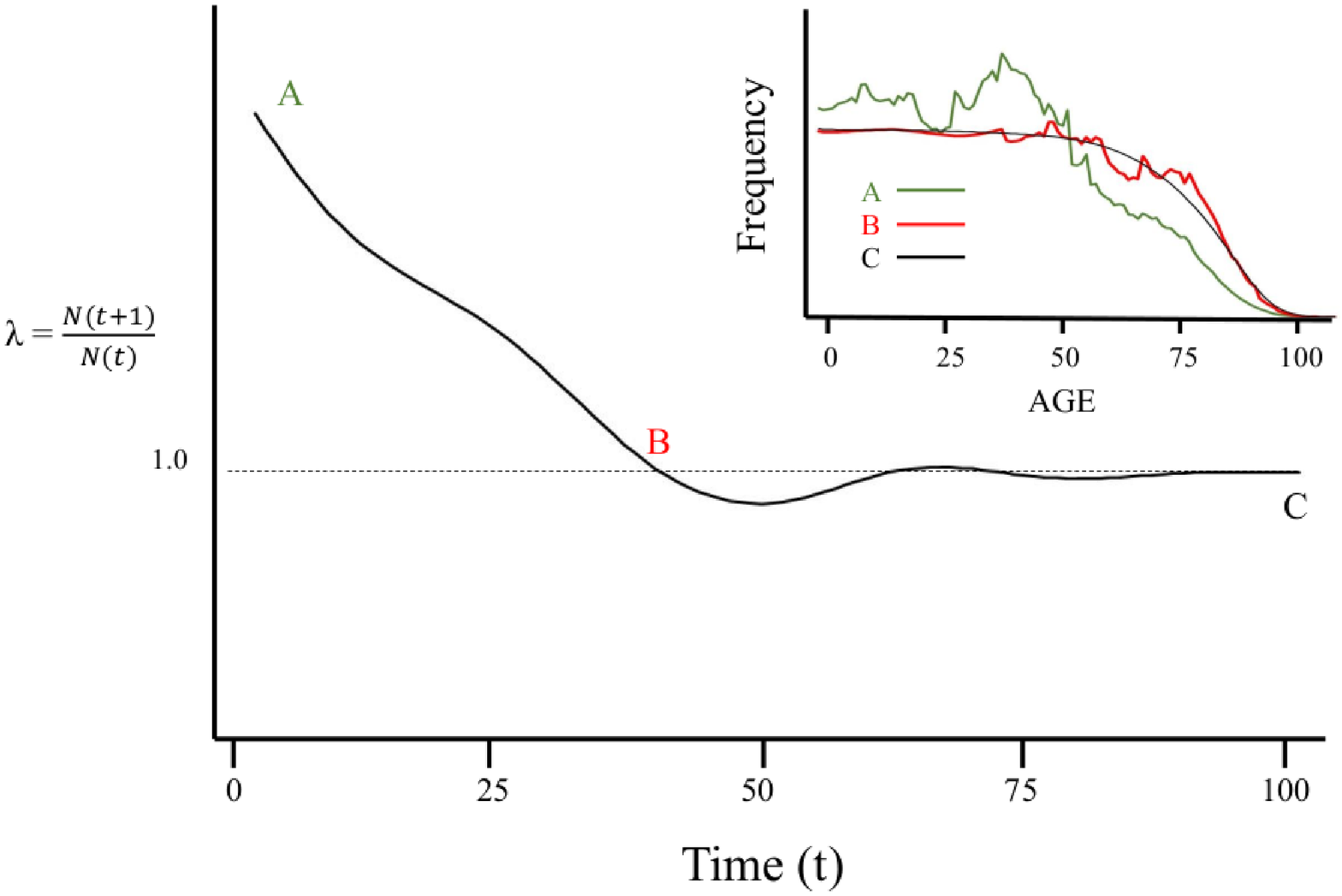}

\caption{\label{fig3}Trajectory of growth rate and age structure in a hypothetical
population starting from a positive rate and converging to zero growth
(stationarity). The initial age structure was based on the U.S. population
in 2000, the fertility rates on a standard age-specific fertility
schedule in humans scaled to replacement-level (=1.0) and survivorship
on the female rates in 2006 {[}7{]}. N(t+1)/N(t) is the ratio of number
in the population at time t+1 to the number at time and is denoted
\textgreek{l}. Frequency in the inset refers to the age distribution
of the population. Points labeled A, B, and C correspond to the starting
growth rates, the point at which growth rate first reaches replacement
level (i.e., transient stationarity), and the point at which growth
rate is constant at zero (i.e., fixed stationarity), respectively.
Replacement levels of growth required approximately 40 years from
the start (i.e., A-to-B) and another 60 years to become fixed (i.e.,
B-to-C). Note the small oscillations around stationarity after B as
the age structure converges to C.}

\end{figure}

\section{\textbf{Oscillatory stationary population identity and It's Amplitude}}

Suppose that the population remains stationary for an infinitely long
period of time, except in shorter time intervals in between due to
perturbations. After small perturbations in the population due to
vital events population deviates temporarily to a non-stationary state
for a brief time-period before restoring stationary properties. When
the population is stationary then we know that SPI holds (see for
example \cite{RaoCarey2015}). SPI holds here we mean that it is true
for all ages, i.e., the proportion of the population who are at age
$a$\emph{ units }is same as the proportion of the population who
have $a$\emph{ units }remaining for each age $a.$ Suppose the population
remains stationary during the interval $\left[0,t\right]$ and let
there be a vital event during the interval $(t,t+\delta)$ for a positive
$\delta$ which is very very small. Then during $(t,t+\delta)$ SPI
(in a strict sense) is not true and SPI remains not true until population
remains strictly non-stationary (say until $\delta_{1}$ for $\delta_{1}>\delta$).
As soon as stationarity is restored SPI will be true again until the
next vital event. There will be finite or infinite \emph{cycles} of
\emph{stationary to non-stationary to stationary populations} and
hence SPI is true intermittently. 
\begin{defn}
\label{def: predominant Stationary}We define a population as a \emph{predominantly
a stationary population} if 

\begin{equation}
[t_{0},t_{1})\cup\stackrel[M,i=1]{\infty}{\bigcup}[\delta_{i},t_{i+1})>\stackrel[N,i=1]{\infty}{\bigcup}[t_{i},\delta_{i}),\label{eq: predominant stationary}
\end{equation}

where SPI holds for the disjoint collection of intervals $M=\left\{ [t_{0},t_{1}),[\delta_{1},t_{2}),[\delta_{2},t_{3}),\cdots\right\} ,$
and SPI does not hold for the disjoint collection of intervals $N=\left\{ [t_{1},\delta_{1}),[t_{2},\delta_{2}),\cdots\right\} .$
The L.H.S. of (\ref{eq: predominant stationary}) is the union over
$i$ for the time intervals corresponding to the set $M$ and the
R.H.S. of (\ref{eq: predominant stationary}) is the union over $i$
values for the time intervals corresponding to the set\textbf{ $N.$}
If 

\begin{equation}
[t_{0},t_{1})\cup\stackrel[M,i=1]{\infty}{\bigcup}[\delta_{i},t_{i+1})<\stackrel[N,i=1]{\infty}{\bigcup}[t_{i},\delta_{i}),\label{eq:predom non-stat}
\end{equation}

then we define a population as \emph{predominantly a non-stationary}.
A population is neither \emph{predominantly a stationary population
nor} \emph{predominantly a non-stationary if}\textbf{\emph{ }}

\begin{equation}
[t_{0},t_{1})\cup\stackrel[M,i=1]{\infty}{\bigcup}[\delta_{i},t_{i+1})=\stackrel[N,i=1]{\infty}{\bigcup}[t_{i},\delta_{i}).\label{eq:define =00003D case}
\end{equation}
\end{defn}
We define oscillatory property of SPI as follows:
\begin{defn}
We define a criteria that the SPI is \emph{oscillatory} on $M$ with
\emph{uniform amplitude} whenever the following statement is true:
\end{defn}
\begin{equation}
\left\{ t_{1}=\frac{t_{0}+\delta_{1}}{2},t_{2}=\frac{\delta_{1}+\delta_{2}}{2},t_{3}=\frac{\delta_{2}+\delta_{3}}{2},...\right\} .\label{eq:statement}
\end{equation}

\begin{thm}
\textbf{\label{conj:Oscillatory-Carey's-Equality.}} For a predominantly
stationary population defined in the Definition {[}\ref{def: predominant Stationary}{]},
SPI exists except for shorter intermittent intervals when population
is non-stationary. 
\end{thm}
\begin{proof}
Let the population be stationary during $[t_{0},t_{1})$ and a small
perturbation (vital event(s)) takes place at $t_{1}$ such that the
population deviates from stationary properties. Suppose there is a
vital event(s) (at the time $\delta_{1}$ for some $\delta_{1}>t_{1}$),
which balances deviated stationary population back to stationary mode.
Suppose at time $t_{2}$ for $t_{2}>\delta_{1}$ the population again
deviates from stationary mode due to vital event(s) and gets restored
at time $\delta_{2}$ for $\delta_{2}>t_{2}$ such that the population
remains non-stationary in the interval $[t_{2},\delta_{2}).$ Suppose
this cycle of stationary population to non-stationary and back to
stationary population continues to repeat at different time points
$t$ and $\delta.$ SPI holds for the disjoint collection of intervals,
$M$ and does not hold for the disjoint collection of intervals, $N.$
Because the population is predominantly stationary, the following
inequality holds
\begin{eqnarray}
\frac{[t_{0},t_{1})\cup\stackrel[M,i=1]{\infty}{\bigcup}[\delta_{i},t_{i+1})}{\stackrel[N,i=1]{\infty}{\bigcup}[t_{i},\delta_{i})} & > & 1.\label{eq:ineq1}\\
\nonumber 
\end{eqnarray}
\end{proof}
Let us denote $\mathcal{S}_{0}=\frac{[t_{0},t_{1})\cup\stackrel[M,i=1]{\infty}{\bigcup}[\delta_{i},t_{i+1})}{\stackrel[N,i=1]{\infty}{\bigcup}[t_{i},\delta_{i})}.$ 

Whenever $\left[\Sigma_{i=1}^{\infty}(\delta_{i}-t_{i})\mbox{ for }\delta_{i},t_{i}\in N\right]$
$>\left[(t_{1}-t_{0})+\Sigma_{i=1}^{\infty}(t_{i+1}-\delta_{i})\mbox{ for }\delta_{i},t_{i}\in M\right]$,
then $\mathcal{S}_{0}<1$. 

We call this property of holding and not holding SPI over disjoint
intervals constructed in the proof of the Theorem \ref{conj:Oscillatory-Carey's-Equality.}
as the \emph{O-SPI}. We associate the idea of amplitude with the length
of time when SPI holds. The \emph{amplitudes} of SPI are defined here
as the lengths of the intervals of the set $M.$ 
\begin{thm}
\label{thm:6}Given a finite time set-up of disjoint intervals of
$M$ and $N$ up to $[\delta_{k},t_{k+1})\in M.$ If SPI is \emph{oscillatory}
on $M$ \emph{with uniform amplitude,} \emph{then} $\mathcal{S}_{0}=1$
but converse need not be true.
\end{thm}
\begin{proof}
When SPI is oscillatory on $M$ with uniform amplitude then the statement
(\ref{eq:statement}) is true. Hence we can see that

$\left[\Sigma_{i=1}^{\infty}(\delta_{i}-t_{i})\mbox{ for }\delta_{i},t_{i}\in N\right]$
$=\left[(t_{1}-t_{0})+\Sigma_{i=1}^{\infty}(t_{i+1}-\delta_{i})\mbox{ for }\delta_{i},t_{i}\in M\right]$,
which implies $\mathcal{S}_{0}=1.$

Conversely, suppose $\mathcal{S}_{0}=1.$ Let us consider events up
to time $t_{k+1}$ in the interval $[t_{0},\infty)$. Let $M_{t_{k+1}}$and
$N_{t_{k+1}}$ be the sub-collection of intervals of $M$ and $N$,
respectively and are given by,

\begin{eqnarray}
M_{t_{k+1}} & = & \left\{ [t_{0},t_{1}),[\delta_{1},t_{2}),[\delta_{2},t_{3}),\cdots,[\delta_{k},t_{k+1})\right\} ,\nonumber \\
N_{t_{k+1}} & = & \left\{ [t_{1},\delta_{1}),[t_{2},\delta_{2}),\cdots,[t_{k},\delta_{k})\right\} .\label{eq:truncated collection}
\end{eqnarray}

For $k=1,$ we have,

\begin{eqnarray}
M_{t_{2}} & = & \left\{ [t_{0},t_{1}),[\delta_{1},t_{2})\right\} ,\nonumber \\
N_{t_{2}} & = & \left\{ [t_{1},\delta_{1})\right\} .\label{eq:truncated up to t2}
\end{eqnarray}

Since $\mathcal{S}_{0}=1,$ from (\ref{eq:truncated up to t2}), we
will have $t_{1}=\frac{t_{0}+2\delta_{1}-t_{2}}{2}.$ Since $t_{2}>\delta_{1}$,
we have

\begin{equation}
\frac{t_{0}+2\delta_{1}-t_{2}}{2}\neq\frac{t_{0}+\delta_{0}}{2}.\label{eq:NEQ}
\end{equation}

The inequality (\ref{eq:NEQ}) indicates there is no uniform amplitude.
\end{proof}
\begin{cor}
\emph{For a predominantly non-stationary population SPI may hold even
in small intermittent intervals.}
\end{cor}

\section{\textbf{Connection between main theorems, graphical results and applications}}

Several new properties and implications of SPI were proved in this
article through various Theorems {[}Theorem \ref{thm1}, Theorem \ref{thm2},
Theorem \ref{conj:Oscillatory-Carey's-Equality.}, Theorem \ref{thm:6},
Theorem \ref{thm9} and Theorem \ref{theorem:inequality}{]} and Lemma
\ref{lemma in Appendix}. These theorems which take implications of
SPI to different directions (for example, Theorem \ref{thm1}, Theorem
\ref{thm2}, Theorem \ref{conj:Oscillatory-Carey's-Equality.}, Theorem
\ref{thm:6}) show newer avenues of the interface of SPI between stationary
and non-stationary populations and O-SPI. Constructions of captive
cohorts and logic of matching the duration of LL and LR developed
in Carey - Rao Theorem on SPI \cite{RaoCarey2015} helped us to prove
arguments in Theorem \ref{thm1}. This theorem implies that when the
fraction of the population at age $a$ is not equal to the fraction
of the population whose remaining years to live is $a$ for some age
$a$, then the population could be either stationary or non-stationary.
For a stationary population shown in Figure \ref{Stationary-population-color},
these fractions are equal at all ages $a$ or at all age groups if
group-wise fractions are considered. When these fractions are not
equal for each age $a$ then the population is non-stationary. Therefore,
Theorem \ref{thm1} helped us investigate properties of the SPI that
interface stationary and non-stationary populations. 

An example of the relationship between the percentage of a population
that have lived x years and the percentage of persons who have x years
remaining in a stationary population is visualized in Figure \ref{Stationary-population-color}
for a hypothetical population based on the 2006 U.S. female life table.
This graphic shows that there are 14.8\% of this stationary population
that are aged 70 to 90 years old and, at the same time, there are
14.8\% of the population who have from 70 to 90 years remaining. These
percentages for remaining years are distributed among the younger
individuals from newborn (x=0), most of whom have from 70 t 90 years
remaining to a miniscule percentage of 40 year-olds who will live
another 70 years (i.e., to super-centenarians, age 110). Theorem \ref{thm2}
is the first step towards specifying the behavior of SPI on partitions
of stationary and non-stationary sub-populations of the total population.
This implies, when partitioning of the total population is done into
a collection of stationary sub-populations then the aforesaid fractions
remain equal in each of the sub-population. Theorem \ref{thm2} also
implies that when a population is partitioned into sub-populations,
these fractions need not be equal if stationary principles are not
preserved (see Figure \ref{fig2}). 

The relationship of the fraction of individuals in a population that
have lived x years (i.e., the age distribution) relative to the fraction
of individuals in the same population that have x years remaining
is shown in Figure 2 for two hypothetical populations, one with a
negative (r=-0.01) growth rate (Fig. 2a) and another with a positive
(r=0.01) growth rate (Fig. 2b). Each is shown relative to the stationary
(r=0.00) case. Several aspects of this figure merit comment. First,
note the equivalency of the fraction of the population that have lived
x years and the fraction that have x years to live as shown in the
curves for the stationary population. In other words, the life-lived/life-left
curves are superimposed. In contrast the trajectories for the life-lived
and the life-left curves for populations with either negative or positive
growth rates are separate as shown by the departure of the dashed
and solid black lines in each graph. Note in the top graph for a population
with the negative growth rate the fraction of the population that
are young (e.g., 0 to 20 years is low relative the fraction that are
old (e.g., 60 to 80 years). In other words, the age structure of decreasing
populations is skewed to the older age classes. However, because population
is older, the fraction of individuals with fewer number of years to
live (e.g., \textless 20) is higher relative to the fraction of persons
who have many years left to live (e.g., \textgreater 60 more years
remaining). The exact opposite relationship of LL to LR is evident
in a population a positive growth rate as shown in the bottom graph
of Figure \ref{fig2}. That is, the skew towards the young in a growing
population results in a skew toward the fraction of individual (i.e.,
the young) who have many years remaining.

Through Theorem \ref{conj:Oscillatory-Carey's-Equality.}, we have
shown for a stationary population how SPI could hold in alternate
time intervals. Within the construction of oscillatory properties,
we have introduced amplitudes of SPI which provides the lengths of
time intervals for which the SPI holds. We have introduced the idea
of O-SPI which over the time will have practical applications in understanding
population dynamics through switching of stationary and non-stationary
populations (See both Figure \ref{fig3} and next paragraph). The
concept of \emph{transient stationarity} is visualized in Figure \ref{fig3}
for a population converging from a positive growth rate to a fixed
(replacement-level) stationary state. This figure shows the change
in growth rate, $\lambda$, in the main graphic and the age structure
of the population that corresponds to three different points (A, B,
and C) along this growth trajectory. The age structure (inset) at
$t=0$ corresponds to a growing population with a bi-modal distribution,
one mode from birth to age 25 and the other from 25 to 50 years. This
corresponds to point A in the main graphic (rapid growth rate). At
around $t=40$ the population growth rate, \textgreek{l}, had decreased
to zero (Point B). However, this was a transient condition because
the age structure (shown in the inset) was not yet stable. This \textquotedblleft transiency\textquotedblright{}
in growth rate and age structure continued until both had converged
to fixed stationarity, a state corresponding to point C in the main
graphic showing changes in $\lambda$ and in the inset showing the
age distribution. Connections between the properties of O-SPI and
countries or populations with net reproduction rates (NRR) around
the value one ($1.0$) can be investigated using the properties proved
in this article. When the intrinsic growth rates are highly dynamic
in the populations then achieving the net reproductive rates around
the value one may not stay for a longer period, and the duration of
the time for which the status of \textasciiacute NRR = $1.0$\textasciiacute{}
in the population might be very short lived. Implications of \textasciiacute NRR
= $1.0$\textasciiacute{} and properties of O-SPI across several populations
can be studied to understand long-range population dynamics. 

Since every human population has an underlying life table, every human
population can form the basis of a model stationary population \cite{Preston2011}.
Therefore it follows that understanding the deeper properties of stationary
populations as described here and elsewhere \cite{Ryder1973,MullerCarey2004,RaoCarey2015}
will add important depth to population theory more generally. Second,
understanding the oscillatory behavior of populations as they approach
stationarity is important inasmuch as this behavior is tightly linked
to the concept of population momentum\textemdash the continuation
of growth after a population has achieved replacement-level fertility
\cite{keyfitz1971,Schoen2003,Rao-Notices}. Momentum and population
aging are essentially two aspects of the same phenomenon \cite{KimSchoenSarma},
and momentum is likely to be responsible for most of the future growth
in the world\textquoteright s population \cite{Bongaarts1999,Cohen1995}.
Therefore, a deeper understanding of the underlying dynamics of population
stationarity, momentum, and convergence and concepts concerned with
the demographic transition will strengthen the foundations for the
development of sound population policy including family planning,
aging, and social security.

\section{\textbf{Discussion}}

The number of years different individuals have lived in a population,
as well as the number of years these individuals have left, are universal
properties of all populations. Whereas the first is a static characteristic
of populations inasmuch as it specifies age structure, the second
is a dynamic concept since it designates the future population\textquoteright s
actuarial properties. This second property is more complex than the
first inasmuch as it describes distributions within a distribution
i.e., the allocation of individual deaths within each of the 100+
age groups of the age distribution. Both of these population characteristics
are important in both basic and applied demographic contexts. The
first property is concerned with the relationship of different population
age groups (e.g. dependency ratios; population aging) and the second
is concerned with future deaths (e.g. how many deaths will occur in
the next $1,2$, or $5$ years). Since the age structure of a population
must logically be connected to its future death distribution, the
implicit qualitative relationship between life-years lived and life-
years is both obvious and intuitive. However, the explicit quantitative
relationship between life-lived and left was neither obvious nor intuitive
prior to the discovery of the SPI. Because of the importance of linking
the actuarial properties of populations with their age structure as
SPI does, it follows that exploring this identity in greater mathematical
depth has the potential to provide important new insights into these
linkages in two mathematical contexts. The first is within stationary
populations as we did with the three main properties (Theorems \ref{thm1}
to \ref{conj:Oscillatory-Carey's-Equality.}), and the second context
is between stationary and non-stationary populations as we did with
what we refer to as O-SPI. We still feel the beauty of SPI in population
dynamics is under explored, and the results presented here can be
seen as a step towards a larger goal of understanding non-stationary
populations through such lens. 

\section*{\textbf{Acknowledgments}}

We thank three referees for their critical review of the concepts
introduced and for their several useful structural comments which
helped us to thoroughly revise our original submission. Research supported
in part through the UC Berkeley CEDA grant P30AG0128 to JRC.

\section*{\textbf{Appendix}}

Let,

\begin{eqnarray*}
I & = & [t_{0},t_{1})\bigcup\stackrel[M,i=1]{k}{\bigcup}[\delta_{i},t_{i+1}),\\
J & = & \stackrel[N,i=1]{k}{\bigcup}[t_{i},\delta_{i}).
\end{eqnarray*}
and let $\dot{I}$ and $\dot{J}$ be the partitions of $I$ and $J$,
which are written as,

\textbf{
\[
\dot{I}=\left\{ \left(I_{M}(t_{i})\right):i=1,2,\cdots,k+1\right\} =\left\{ \left(I_{M}(t_{i})\right)\right\} _{i=1}^{k}
\]
} and

\textbf{
\[
\dot{J}=\left\{ \left(J_{N}(t_{i})\right):i=1,2,\cdots,k\right\} =\left\{ \left(J_{N}(t_{i})\right)\right\} _{i=1}^{k},
\]
}where $I_{M}(t_{1})=[t_{0},t_{1}),$ $I_{M}(t_{i})=[\delta_{i},t_{i+1})$
for $i=2,\cdots,k$ and $J_{N}(t_{i})=[t_{i},\delta_{i})$ for $i=1,2,\cdots,k.$
Since $I_{M}(t_{i})$ and $J_{N}(t_{i})$ are non-degenerate intervals,
the lengths of $I_{M}(t_{i})$ and $J_{N}(t_{i})$ are always positive.
Hence, $\text{max}I_{M}(t_{i})$, $\min I_{M}(t_{i})$, $\max J_{N}(t_{i})$,
and $\min J_{N}(t_{i})$ exists. Let $f(a,t_{i})$ be the function
specifying the proportion of individuals at age $a\in A$ during $I_{M}(t_{i})$
for $f(a,t_{i}):I_{M}(t_{i})\rightarrow\mathbb{R}^{+}$ and $A$ be
the set of all ages in the population. Since SPI holds in $I_{M}(t_{i})$,
we have 

\begin{eqnarray}
Prob\left[f(a,t_{i})=g(a,t_{i})\forall a,t_{i}\right] & = & 1\mbox{ if }f(a,t_{i}):I_{M}(t_{i})\rightarrow\mathbb{R}^{+}\nonumber \\
 & = & 0\mbox{ otherwise,}\label{eq:CE Prob}
\end{eqnarray}
where $g(a,t_{i})$ is the function specifying remaining LR at age
$a$ during $I_{M}(t_{i}).$
\begin{lem}
\label{lemma in Appendix}Suppose $\triangle f(a,t_{i})=\hat{f}(a,t_{i})-\check{f}(a,t_{i})$
for $i=1,2,3,...,k$, where $\hat{f}(a,t_{i})=\max_{a}f(a,t_{i})$
and $\check{f}(a,t_{i})=\min_{a}f(a,t_{i})$, then $\triangle f(a,t_{i})$
is bounded for each $I_{M}(t_{i}).$ 
\end{lem}
\begin{proof}
If there are at least two age groups in $A$, then $\hat{f}(a,t_{i})$
and $\check{f}(a,t_{i})$ exists within $I_{M}(t_{i})$ and they are
distinct. Suppose there are only two age groups in $A$, then (\ref{eq:CE Prob})
guarantees that there exist $\hat{g}(a,t_{i})$ and $\check{g}(a,t_{i})$
for $\hat{g}(a,t_{i})=\max_{a}g(a,t_{i})$ and $\check{g}(a,t_{i})=\min_{a}g(a,t_{i})$.
This implies, $\triangle f(a,t_{i})<\hat{g}(a,t_{i})+\check{g}(a,t_{i}).$
This inequality follows even if there are more than two age groups
in $A$, hence $\triangle f(a,t_{i})$ is bounded. 
\end{proof}
\begin{thm}
\label{thm9}$\frac{1}{\Sigma\triangle f(a,t_{i})}$ is bounded on
$\left[t_{0},t_{k+1}\right).$
\end{thm}
\begin{proof}
Since $\triangle f(a,t_{i})>0$ and $\triangle f(a,t_{i})$ is bounded
on $I_{M}(t_{i})$ by the Lemma (\ref{lemma in Appendix}), the result
follows. 
\end{proof}
Suppose $\hat{f}(a,t_{i})$ is concentrated around mean age of the
population and $\check{f}(a,t_{i})$ is concentrated around the very
old age of the population, then $\triangle f(a,t_{i})$ is an increasing
function indicates one or more of the following three situations;
i) longevity of the population is increasing without much change in
the mean age, ii) mean age is reducing without reducing in longevity,
iii) mean age is reducing and simultaneously longevity is increasing. 
\begin{thm}
\label{theorem:inequality}Suppose the partitions $\dot{I}$ and $\dot{J}$
are given, then $1+\frac{k^{2}}{\hat{f}(a,t_{i})+\check{f}(a,t_{i})}>k\left(\frac{1}{\hat{f}(a,t_{i})}+\frac{1}{\check{f}(a,t_{i})}\right)$.
\end{thm}
\begin{proof}
Consider the expression

\begin{equation}
\left(\check{f}(a,t_{i})-\sum_{i=1\mbox{ for }t_{i}\in I}^{\infty}\int_{0}^{\infty}f(a,t_{i})da\right)\left(\hat{f}(a,t_{i})-\sum_{i=1\mbox{ for }t_{i}\in I}^{\infty}\int_{0}^{\infty}f(a,t_{i})da\right).\label{eq:product-integrals}
\end{equation}

Since $\sum_{i=1}^{\infty}\int_{0}^{\infty}f(a,t_{i})da=k$ and both
the terms of the expression (\ref{eq:product-integrals}) are negative,
(\ref{eq:product-integrals}) can be written as 

\begin{equation}
\left(\check{f}(a,t_{i})-k\right)\left(\hat{f}(a,t_{i})-k\right)>0.\label{eq:products}
\end{equation}

Simplifying (\ref{eq:products}) we will obtain desired result. 
\end{proof}
\begin{rem}
\label{remark after Thm}For each $t_{i}$ for $i=1,2,\cdots,k$,
without taking the summations in (\ref{eq:product-integrals}), we
have
\end{rem}
\[
\left(\check{f}(a,t_{i})-f(a,t_{i})\right)\left(\hat{f}(a,t_{i})-f(a,t_{i})\right)\begin{cases}
\begin{array}{cc}
=0 & \mbox{if }\hat{f}(a,t_{i})=f(a,t_{i})\mbox{ or }\check{f}(a,t_{i})=f(a,t_{i})\\
<0 & \mbox{if }\mbox{ }\check{f}(a,t_{i})<f(a,t_{i})<\hat{f}(a,t_{i})
\end{array}\end{cases}
\]

and

\[
\left(\check{f}(a,t_{i})-\int_{0}^{\infty}f(a,t_{i})da\right)\left(\hat{f}(a,t_{i})-\int_{0}^{\infty}f(a,t_{i})da\right)>0.
\]

Let $\varphi(a,t_{i})$ be the function specifying the proportion
of individuals at age $a\in B$ during $J_{N}(t_{i})$ for $\varphi(a,t_{i}):J_{N}(t_{i})\rightarrow\mathbb{R}^{+}$
and $B$ be the set of all ages in the population when SPI does not
hold. Suppose $\hat{\varphi}(a,t_{i})=\max_{a}\varphi(a,t_{i})$ and
$\check{\varphi}(a,t_{i})=\min_{a}\varphi(a,t_{i})$. We note that,
equivalent versions of Theorem \ref{theorem:inequality} and Remark
\ref{remark after Thm} for the age functions $\varphi,$ $\hat{\varphi}(a,t_{i}),$
$\check{\varphi}(a,t_{i})$ still hold. Under the continuous transition
of decreasing population sizes over the interval $[t_{0},t_{k+1}),$
let us assume $\hat{f}(a,t_{1})>\hat{f}(a,t_{2})>\cdots>\hat{f}(a,t_{k+1})$
and $\hat{\varphi}(a,t_{1})>\hat{\varphi}(a,t_{2})>\cdots>\hat{\varphi}(a,t_{k}).$
This implies, $\hat{f}(a,t_{1})>\hat{\varphi}(a,t_{1})>\cdots>\hat{\varphi}(a,t_{k})>\hat{f}(a,t_{k+1})$.
Also, $\int_{0}^{\infty}f(a,t_{1})da-\hat{f}(a,t_{1})<\int_{0}^{\infty}\varphi(a,t_{1})da-\hat{\varphi}(a,t_{1})<\cdots<\int_{0}^{\infty}\varphi(a,t_{k})da-\hat{\varphi}(a,t_{k})<\int_{0}^{\infty}f(a,t_{k+1})da-\hat{f}(a,t_{k+1}),$
and this leads to $1-\hat{f}(a,t_{1})<1-\hat{\varphi}(a,t_{1})<\cdots<1-\hat{\varphi}(a,t_{k})<1-\hat{f}(a,t_{k+1}).$
We can model the dynamics of these maximum and minimum fractions over
the time period using the following logistic growth models with certain
limiting points of these fractions.

\begin{eqnarray}
\frac{d\hat{f}(a,t)}{dt} & = & r_{1}\hat{f}(a,t)\left(1-\frac{\hat{f}(a,t)}{\left(\hat{f}(a,t)\right)_{e}}\right)\label{eq:ode1}\\
\frac{d\hat{\varphi}(a,t)}{dt} & = & r_{2}\hat{\varphi}(a,t)\left(1-\frac{\hat{\varphi}(a,t)}{\left(\hat{\varphi}(a,t)\right)_{e}}\right)\label{eq:ode2}\\
\frac{d\check{f}(a,t)}{dt} & = & r_{3}\check{f}(a,t)\left(1-\frac{\check{f}(a,t)}{\left(\check{f}(a,t)\right)_{e}}\right)\label{eq:ode3}\\
\frac{d\check{\varphi}(a,t_{i})}{dt} & = & r_{4}\check{\varphi}(a,t_{i})\left(1-\frac{\check{\varphi}(a,t_{i})}{\left(\check{\varphi}(a,t_{i})\right)_{e}}\right),\label{eq:ode4}
\end{eqnarray}

where $r_{1},r_{2},r_{3}$ and $r_{4}$ are rates of declines in maximum
and minimum fractions and $\left(\hat{f}(a,t)\right)_{e},$ $\left(\hat{\varphi}(a,t)\right)_{e},$
$\left(\check{f}(a,t)\right)_{e},$ $\left(\check{\varphi}(a,t_{i})\right)_{e}$
are limiting points of the fractions $\hat{f}(a,t),$ $\hat{\varphi}(a,t),$
$\check{f}(a,t),$ $\check{\varphi}(a,t_{i})$, respectively. Further
we provide partial differential equations models by treating $\hat{f}(a,t),$
$\hat{\varphi}(a,t),$ $\check{f}(a,t),$ $\check{\varphi}(a,t_{i})$
as continuous variables. First we consider two pairs of variables
$\left\{ \hat{f}(a,t),\hat{\varphi}(a,t)\right\} $, $\left\{ \check{f}(a,t),\check{\varphi}(a,t_{i})\right\} $
and corresponding dependent variables $u_{1}\left(\hat{f}(a,t),\hat{\varphi}(a,t)\right)$,
$u_{2}\left(\check{f}(a,t),\check{\varphi}(a,t_{i})\right)$ to build
two models (\ref{eq: pde1}) and (\ref{eq: pde2}). These two models
provide dynamics of simultaneous occurrences of stationary and non-stationary
populations. If we want to follow dynamics of $\hat{f}$ and $\hat{\varphi}$
on the time interval $[t_{0},t_{\infty})$ by considering two pairs
of independent variables $\left\{ t,\hat{f}(a,t)\right\} $, $\left\{ t,\hat{\varphi}(a,t)\right\} $
with corresponding dependent variables $v_{1}\left(t,\hat{f}(a,t)\right)$,
$v_{1}\left(t,\hat{\varphi}(a,t)\right)$ then the PDE models we considered
are given in (\ref{eq:pde3}) and (\ref{eq:pde4}). Here $\tau_{1}$
and $\tau_{2}$ are constants, which could indicate speed of the dynamics
of peaks of the maximum fractions. Similarly, dynamics of $\check{f}$
and $\check{\varphi}$ with dependent variables $w_{1}\left(t,\check{f}(a,t)\right)$
and $w_{2}\left(t,\check{\varphi}(a,t_{i})\right)$ are modeled as
per equations given in (\ref{eq:pde5}) and (\ref{eq:pde6}), where
$\tau_{3}$ and $\tau_{4}$ are constants indicate speed with which
these variables move. 

\begin{eqnarray}
\frac{\partial u_{1}\left(\hat{f}(a,t),\hat{\varphi}(a,t)\right)}{\partial\hat{f}} & = & -\hat{\varphi}\frac{\partial u_{1}\left(\hat{f}(a,t),\hat{\varphi}(a,t)\right)}{\partial\hat{\varphi}}\label{eq: pde1}\\
\frac{\partial u_{2}\left(\check{f}(a,t),\check{\varphi}(a,t_{i})\right)}{\partial\check{f}} & = & -\check{\varphi}\frac{\partial u_{2}\left(\check{f}(a,t),\check{\varphi}(a,t_{i})\right)}{\partial\check{\varphi}}\label{eq: pde2}\\
\frac{\partial v_{1}\left(t,\hat{f}(a,t)\right)}{\partial t} & = & -\tau_{1}\frac{\partial v_{1}\left(t,\hat{f}(a,t)\right)}{\partial\hat{f}},\label{eq:pde3}
\end{eqnarray}

\begin{eqnarray}
\frac{\partial v_{2}\left(t,\hat{\varphi}(a,t)\right)}{\partial t} & = & -\tau_{2}\frac{\partial v_{2}\left(t,\hat{\varphi}(a,t)\right)}{\partial\hat{\varphi}}\label{eq:pde4}\\
\frac{\partial w_{1}\left(t,\check{f}(a,t)\right)}{\partial t} & = & -\tau_{3}\frac{\partial w_{1}\left(t,\check{f}(a,t)\right)}{\partial\check{f}}\label{eq:pde5}\\
\frac{\partial w_{2}\left(t,\check{\varphi}(a,t)\right)}{\partial t} & = & -\tau_{4}\frac{\partial w_{4}\left(t,\check{\varphi}(a,t)\right)}{\partial\check{\varphi}}.\label{eq:pde6}
\end{eqnarray}
Diffusion type of equations appear in several situations of modeling
in biology, for example refer to the book \cite{Benoit-book}. Further
applications of diffusion type of equations appear in studying growth
of cell populations, see \cite{Boulanour-EJDE,Lebowitz-jmb,Rotternberg}.

\end{document}